\documentclass{article}
\usepackage{graphicx}
\usepackage{amsfonts}
\usepackage{amsmath}
\usepackage{amssymb}
\usepackage{url}

\usepackage{dsfont}
\usepackage{fancyhdr}
\usepackage{indentfirst}
\usepackage{enumerate}
\usepackage[normalem]{ulem}
\usepackage{mathrsfs}
\usepackage{multirow}
\usepackage{diagbox}
\usepackage[misc]{ifsym}

 \usepackage[colorlinks=true,citecolor=blue]{hyperref}
\usepackage{amsthm}
\usepackage{color}

\definecolor{DarkGreen}{rgb}{0.2,0.6,0.2}

\definecolor{purple}{rgb}{0.6,0.3,0.8}

\usepackage{natbib}
\usepackage{comment}

\def\d{\mathrm{d}}

\newcommand{\R}{\mathbb{R}}

\newcommand{\N}{\mathbb{N}}

\newcommand{\M}{\mathcal{M}}

\newcommand{\id}{\mathds{1}}

\renewcommand{\ge}{\geqslant}

\renewcommand{\geq}{\geqslant}
\renewcommand{\leq}{\leqslant}
\renewcommand{\epsilon}{\varepsilon}

\theoremstyle{plain}
\newtheorem{theorem}{Theorem}

\newtheorem{proposition}{Proposition}
\theoremstyle{definition}
\newtheorem{definition}{Definition}
\newtheorem{example}{Example}

\theoremstyle{remark}
\newtheorem{remark}{Remark}

\usepackage[onehalfspacing]{setspace}

\newcommand{\thedate}{\today}

\usepackage{footmisc}
\begin{document}
\title{A duality between utility transforms and probability distortions}

\author{Christopher P. Chambers\thanks{Department of Economics, Georgetown University, USA.  \Letter~\href{mailto:cc1950@georgetown.edu}{cc1950@georgetown.edu}.} \and Peng Liu\thanks{School of Mathematics, Statistics and Actuarial Science, University of Essex, UK.  \Letter~\href{mailto:peng.liu@essex.ac.uk}{peng.liu@essex.ac.uk}.} \and
  Ruodu Wang\thanks%
  {Department of Statistics and Actuarial Science,
  University of Waterloo, Canada.
 \Letter~\href{mailto:wang@uwaterloo.ca}{wang@uwaterloo.ca}.  RW acknowledges financial support from the Natural Sciences and Engineering Research Council of Canada (RGPIN-2018-03823 and CRC-2022-00141.}  }

  \date{\thedate}

 \maketitle
\begin{abstract}
 In this paper, we establish a mathematical duality between utility transforms and
 probability distortions.
These transforms play a central role in decision under risk by forming the foundation for the classic theories of expected utility, dual utility, and  rank-dependent utility.
Our main results establish that probability distortions are characterized by commutation with utility transforms, and utility transforms are characterized by commutation with probability distortions. These results require no additional conditions,
and hence each class can be axiomatized with only one property.
Moreover, under monotonicity, rank-dependent utility transforms can be characterized by set commutation with either utility transforms or probability distortions.

\begin{bfseries}Key-words\end{bfseries}: Distributional transforms, probability distortions, utility transforms, rank-dependent-utility transforms; quantiles
\end{abstract}

 \section{Introduction}

 Distributional transforms are mappings from one set of probability distributions to another set of distributions.  They are widely used in economics, finance, and risk analysis. A classical example of a distributional transformation is the Lorenz curve (\citet{L05} and \citet{Gas71}).  Formally, we can think of a mapping carrying a distribution of wealth, represented by a cumulative distribution function (cdf), to another cdf over percentages.  The Lorenz curve evaluated at a percentage $p$ specifies the proportion of wealth held by $p$ poorest individuals.  The Lorenz curve has all of the properties of a cdf, and hence can be viewed as a cdf itself, rendering the Lorenz map a distributional transform.   % a simple example, the Lorenz curve (\cite{L05} and \cite{Gas71}) in economics, is a distributional transform.
 The recent study \cite{LSW21} contains a  general treatment of distributional transforms with many other examples.
 Two special classes of  distributional transforms, utility transforms and probability distortions, play a central role in decision theory.  Informally, the former reflects the induced distribution of utils that a given distribution of wealth induces, whereas the latter is a classical representation device used in much of behavioral economics.

% The expected value of the utility of a random final wealth is called the expected utility, the dominant model in decision theory axiomatized by \cite{vNM47}.  It is well understood by every colleague student in economics that a key step in making decision over risks is to transform the value of the wealth level to a utility value, before aggregating them via an expectation.

A very simple case of a decision model based on a distributional transform is the ``dual theory'' of \citet{Y87}.
The individual is supposed to base decisions over a pair of distributions on the expected values of the transformed distributions, rather than the distributions themselves.
In risk management, this mapping is known as a distortion risk measure.  Important special cases include both the Value-at-Risk and the Expected Shortfall, the two regulatory risk measures in banking and insurance (see, e.g., \cite{FS16} and \cite{MFE15}).  In the same context, the distorted distribution can also be used to model tail risk (\cite{LW21}).   \cite{LSW21} axiomatized  probability distortions using three conditions, which we summarize in Section \ref{Sec:PD}.

The rank-dependent utility (RDU) model (\cite{Qui82, Qui93} and \cite{S89}) is one of the most common alternatives to the classical expected utility theory. It effectively generalizes Yaari's dual theory by allowing both attitudes toward risk in the form of a utility index, as well as a distributional transform in the form of what is termed a ``probability distortion.''   Probability distortions are distributional transforms that work via transforming the cumulative probability of receiving some outcome according to some prespecified nondecreasing function.  In this sense, the RDU model can be viewed as a composition of a utility transform and a probability distortion. % and it incorporates the psychological biases of individuals in decision making. 
  Another related approach is \cite{TK92}'s cumulative prospect theory, which can be viewed as a combination of RDU for gains and losses, respectively.

A somewhat more recent decision theory that might also be understood as being based on distributional transforms is the theory of \citet{BB12}, whereby an individual facing probabilistic risk is supposed to maximize an expected utility of some probability less a cost of ``choosing'' that probability.  The cost is typically given as the relative entropy of the chosen probability to the true probability.  This chosen probability can be viewed as a distributional transform.\footnote{Obviously the choice may in general be multi-valued, but any selection would suffice.}  As such, this model is, in a sense, an objective version of the multiplier preferences first axiomatized in economics by \cite{Str11}, themselves a special case of the variational preferences of \cite{MMR06}.

%Instead of considering the expected values of transformed distributions, we isolate the distributional transform from behavioral representations.  We view it as worthy of independent study.  we focus on the distributional transforms themselves involved in the expected utility theory, Yaari's dual utility theory and the RDU theory. The main contribution of this paper is to characterize probability distortions, utility transforms, and RDU transforms as special classes of distributional transforms.

Instead of considering distributional transforms model-by-model, in this work, we isolate the distributional transform from the underlying representations of interest.  Given its ubiquity, we view the distributional transform as worthy of independent study.  Our main goal here is to obtain a complete picture of which types of distributional transforms ``commute'' with respect to other classes of distributional transforms. 
This commutativity is interesting in a few different senses, as we will see later from our results.  First, it has a concrete interpretation as invariance under a certain form of (non-linear) rescaling. Second,
it helps to identify or characterize important classes of decision models.  Third, it allows for convenient operations in applications on these popular distributional transforms. Fourth, and perhaps being the most elegant point of this paper, it offers a new duality between the expected utility theory and the dual theory.

Our first result seeks to understand for which distributional transforms is it the case that an ordinal rescaling of the input distribution results in the same ordinal rescaling of the output distribution.

For example, imagine that an outside observer (an economist) writes down a behavioral model involving distributional transforms, but wants to remain flexible about the timing of their application.  For example, suppose  it is known that a nonlinear income tax is present, but to the economist, only the distribution of pre-tax income is observable.  It may be that the individual facing the risk possesses certain idiosyncratic tax credits or dues, unobservable to the economist.\footnote{It seems reasonable to assume that after-tax income is increasing, perhaps not strictly, with respect to pre-tax income.}  Thus, the individual faces a ``true'' distribution of ex-post income, which is unknown to the economist.  In this case, rather than postulating a distribution over distributions, in the interest of parsimony, it would seem reasonable that the economist apply a distributional transform that results in the ``correct'' transformed distribution of ex-post incomes independently of what the tax rates are.\footnote{This is not to say that the ultimate decision made by the decision maker is independent of the distributional transform---far from it.  Rather, it reflects a model in which the economist understands exactly how probabilities are being transformed, independently of the underyling tax code.}

In this example, the function carrying pre-tax income to post-tax income defines a distributional transform:  a distribution over pre-tax income naturally induces a distribution over post-tax income.  But we can think more generally of nondecreasing transformations.  The transform may represent  ``utils,'' which again may be unobservable to the economist.  Or, perhaps such a transformation represents the discretization of income into a categorical variable, specifying which tax bracket the individual is in, and so forth.  

To this end, we ask the basic question as to which distributional transforms commute with to every distributional transform induced by a nondecreasing function.  In so doing, we provide a minimalist characterization of probability distortions in Section \ref{Sec:PD}: A distributional transform is a probability distortion if and only if it commutes with every such nondecreasing function.  Intuitively, this commutation property means that  a (possibly non-strict) ordinal change in the input distribution leads to the same ordinal change in the output distribution.  This result simplifies and extends the characterization in \cite{LSW21}.\footnote{That paper characterized a related class, but utilizing multiple axioms.  Here, only   one property is used to characterize  probability distortions.}

A natural question then presents itself.  Probability transforms commute with respect to all ordinal changes in the input distribution.  This in itself is a very powerful ``robustness'' condition.  It then seems natural to ask whether they also potentially commute with respect to a \emph{larger} class of transforms.  We would then obtain a broader robustness result for free.  For example, one natural class might be the class of all distributional transforms induced as ``pushforward'' measures for potentially non-monotone but measurable functions.  It turns out that the answer here is negative, in a strong sense.  For any distributional transform which does not arise from a nondecreasing function, there is a probability distortion that does not commute with respect to it.  Thus, the monotone transforms are the largest class of distortions for which we may hope to achieve a natural robustness result for probability distortions.

This result is established in Section \ref{Sec:UT}, where we establish that a distributional transform can be identified as a ``utility transform'' if it commutes with respect to the class of probability distortions.  This commutation property means that a distortion of the input distribution leads to the same distortion of the output distribution.  In a formal sense, these results can be understood as providing a duality between the class of utility transforms and probability distortions.\footnote{The two sets are related via a specific Galois connection, as we discuss in the conclusion.  Galois connections are an order-theoretic notion of duality.}

The combination of   results in Sections \ref{Sec:PD} and \ref{Sec:UT} yields the crucial observation that probability distortions and utility transforms are characterized via commutation with each other.
This observation lends support to the informal idea that Yaari's theory is the (unique) natural dual version of the expected utility theory. %an intuitive and commonly accepted view, but with no concrete mathematical meaning, as far as we are aware.
%Our results thus offer a mathematical duality in the sense that one can mathematically derive one theory from the other, and vice versa.

Finally, let us here discuss a third class of transforms; these are what we call the RDU transforms.  These transforms are the composition of a probability distortion and a utility transform.  We discuss RDU transforms in Section \ref{Sec:RDU}.  These transforms do not in general commute with respect to arbitrary utility transforms.  However, there is a sense in which they do commute.  In particular, suppose we have given an RDU transform, based on a strictly increasing and surjective utility function.  Then it can be shown that for any ordinal rescaling of the input distribution results in a \emph{possibly different} ordinal rescaling of the output distribution.  Similarly, any ordinal rescaling of the output distribution comes from a \emph{possibly different} ordinal rescaling of the input distribution.

This property can be viewed as a commutativity property if we extend the distributional transform to a set-valued mapping.\footnote{In the standard sense, where a function $f$ is extended to sets by $f(A) = \{f(x):x\in A\}$.}  As a set-valued mapping, the set of utility transforms commutes with respect to the distributional transform.

It turns out that this property is characteristic of RDU transforms, under an additional hypothesis of monotonicity with respect to first order stochastic dominance.
%communication with a set of utility transforms, instead of each element of the set, characterizes RDU transforms, under monotonicity.
%Set commutation means that applying a utility transform on the input distribution leads to the application of another utility transform on the output distribution, and vice versa.
This property is clearly weaker than commutation with each utility transform, which requires the same utility transform on the input and the output distributions.
A similar result holds if we replace utility transforms by probability distortions.

We discuss the implications of our results for decision theory in Section \ref{sec:6}.
For the most concise presentation, we focus on compactly supported distributions in the main part of the paper.
Proofs of the main results in Sections \ref{Sec:PD}-\ref{Sec:RDU} are postponed to Appendix \ref{Sec:AA}-\ref{app:D}. Results in Sections \ref{Sec:PD}-\ref{Sec:RDU} are extended to
  other sets of utility and distortion functions in Appendix \ref{app:general} and to
  general   spaces of distributions in Appendix \ref{Sec:AD}.

\section{The Model}
Let $\M$ be the set of compactly supported distributions on $\R$.
%The generalization to more general sets will be studied in Appendix \ref{Sec:AD}.
A distribution in $\M$ will be identified with its cumulative distribution function (cdf).
 For a cdf $F$, we define its left quantile as
\begin{equation*}%\label{quantile}
F_L^{-1}(t)  = \inf\{x\in \R: F(x)\ge t\},~~t\in {(0,1]}, %], ~ \text{and}~ F_L^{-1}(0)=\sup\{x\in\mathbb{R}: F(x)=0\},
\end{equation*}
and its right quantile as
     \begin{align*}%\label{quantile1}
     F_R^{-1}(t)  = \inf\{x\in \R: F(x)> t\},~~t\in  {[0,1)}.
     \end{align*}
By increasing and decreasing, we mean in the non-strict sense.

%For $F, G\in \M$, we say that $F\leq_{\rm st} G$ in \emph{stochastic order} if $F(x)\geq G(x)$ for all $x\in\mathbb{R}$.

A \emph{distortion function (DF)} is an increasing function $d:[0,1]\to [0,1]$ with $d(0)=0$ and $d(1)=1$.
A DF is also called a weighting function (e.g., \cite{TK92}).
The set of all DFs is denoted by $\mathcal F_D$.
A \emph{utility function (UF)} is an increasing and continuous function $u:\R\to \R$.
The set of all UFs is denoted by $\mathcal F_U$.
%Let $\mathcal F_D$ be the set of all increasing functions $d:[0,1]\to [0,1]$ with $d(0)=0$ and $d(1)=1$.
%A function in $\mathcal F_D$ is called a \emph{distortion function (DF)}.
% %$DF$ be the set of all increasing functions $g:[0,1]\to [0,1]$ with $d(0)=0$ and $d(1)=1$.
% Similarly, let $\mathcal F_U$ be the set of all increasing and continuous  functions on $\R$. A function in $\mathcal F_U$ is called a \emph{utility function (UF)}.
 For any increasing function $f$, write $f(x+)=\lim_{y\downarrow x}f(x)$ and $f(x-)=\lim_{y\uparrow x}f(x)$.  %For notational convenience, we extend the domain of $F\in\M$ to $[-\infty,\infty]$ by letting $F(-\infty)=0$ and $F(\infty)=1$. %In the following definition, $\M$ represents a general set of distributions,  e.g., $\M=\M_0$ or $\M=\M$.
\begin{definition}\label{many properties def}
\begin{enumerate}[(i)]
\item  For  $d\in \mathcal F_D$,
{the} \emph{probability distortion} generated by $d$, denoted  by $T_d:\M\to\M$, is defined as $
T_d(F)(x) =(d\circ F)(x+),~x\in\mathbb{R}$.\footnote{For a function $f$, $f(x+)$ denotes the right-hand limit at $x$, which exists in this context.}
\item  For  $u\in \mathcal F_U$,
%For a monotone function $u:\mathbb R\to\mathbb R$,
the \emph{utility transform} $T^u: \mathcal M\to \mathcal M$ is defined as a mapping from the distribution of $X$ to the distribution of $u(X)$, i.e., $T^u(F)=F\circ u^{-1},$ where $F$ is treated as a measure on $\R$, and $u^{-1}(A)=\{x\in\mathbb{R}: u(x)\in A\}$ for any Borel measurable set $A\subseteq\mathbb{R}$.
\item For $T,T':\M \to\M$, we say that $T$ \emph{commutes} with $T'$ if $T\circ T'=T'\circ T$, where $\circ$ denotes composition.
\end{enumerate}
\end{definition}

Denote by $\mathcal{U}$  the set of continuous utility transforms:  $\mathcal{U}=\{T^{u}:u\in \mathcal F_U\}$.  Denote by $\mathcal D$ the set of probability distortions:   $\mathcal D=\{T_d: d\in \mathcal F_D\}$.

\begin{remark}Some caution is required regarding the right limit in Definition \ref{many properties def} (i).  
A simple point is that taking a right limit renders $T_d(F)$ right-continuous,  a requirement for a cdf.
We stress that this is not the same as using a right-continuous version of $d$.
For  $d\in\mathcal F_D$, define $\widehat{d}$ as the right-continuous version of $d$, i.e., $\widehat{d}(x)=d(x+)$ for $x\in [0,1]$. Clearly $\widehat{d}$ may not be in $\mathcal F_D$, if for example $d(0+)>0$.  Even if $\widehat{d}\in \mathcal F_D$, generally, $T_{\widehat{d}}\neq T_d$,  as shown in the following simple example.\end{remark}

\begin{example}  If $d$ is not right-continuous, then there exists $x_0\in (0,1)$ such that $\widehat{d}(x_0)=d(x_0+)>d(x_0)$. Take $F = \mathrm{Bernoulli}(1-x_0)$. Then $F(x)=x_0$ for $x\in [0,1)$. For $x\in [0,1)$, $T_d(F)(x)=(d \circ F)(x+)=d(x_0)<\widehat{d}(x_0)=T_{\widehat{d}}(F)(x)$, showing   $T_{\widehat{d}}\neq T_d$. \end{example}

The above discussion illustrates that, although a right limit is used in Definition \ref{many properties def} (i),
it is not without loss to consider only right-continuous distortion functions.
This subtle difference will be significant in the analysis of general distortion and utility functions, treated in Appendix \ref{app:general}.

%\item We say $T$ is  \emph{monotone} if $T(F)\leq_{\rm st} T(G)$ for  $F\leq_{\rm st} G$,~$F,G\in\mathcal{M}$;

\section{Probability distortions}\label{Sec:PD}

Our first result characterizes   probability distortions as the class of distributional transforms that commute with respect to every utility transform.
 \begin{theorem}\label{PTtheorem1}
For a mapping $T: \M\to\M$,
   $T$ commutes with each element of $\mathcal U$ if and only if $T\in \mathcal D$.
  \end{theorem}

%{\color{red}
%Please see whether the following formulation is better. If so, I suggest changing the entire paper.
%The reason is to also be consistent with Section 5, where I introduce weak commutation.
%Let $\mathrm{U}=\{T^{[g]}:g \in \mathcal F_U\}$
%and $\mathcal F_D=\{T_d: d\in \mathcal F_D\}$, and similarly for others (we should define them properly).
%We say that $T$ commutes with  a set $\mathcal T$ of distributional transforms
%if $T\circ T'= T' \circ T$ for all $T'\in \mathcal T$.

%Theorem: For a mapping $T: \M\to\M$,
 % \com{If this is good, please change throughout. }
%  \begin{enumerate}
 % \item[(i)] $T$ commutes with each element of $\mathcal F_U$ if and only if $T\in \mathcal F_D^*$.  \item[(ii)] $T$ commutes with each element of  $\mathcal F_U^L$  if and only if $T\in \mathcal F_D^R$.
  %\end{enumerate}
  %If the above is used, then I have some further notation suggestions:
  %We should use $\mathcal F_U$ or $\mathcal F_U^*$ for the set of utility functions and $\mathcal F_D$ or $\mathcal F_D^*$ for the set of distortion functions.  Then, I think we should use $u$ for a utility function and $d$ for a distortion function.
  %I do not know how much we need these notations in the main paper. If we make the changes, they should be made throughout.
%}

To interpret the above result, the property requires that only the ``ordinal'' content of the input distribution matters for calculating the distributional transform.  This is in the sense that any ordinal rescaling of the input distribution results in the same ordinal rescaling of the output distribution under the distributional transform.  The result then implies that the distributional transform must necessarily be a probability distribution.

%for a distributional transform $T$, if rescaling the input distribution is the same as rescaling the output distribution with the same utility transform no matter what that utility transform is, then $T$ is a probability distortion.

Theorem \ref{PTtheorem1} is closely related to a characterization result in the literature, Theorem 1 of \cite{LSW21}, which states that for $T:\M\to\M$,
$T$ is monotone, lower-semicontinuous and commuting with each $T^u$ for all $u\in\mathcal F_U^\diamond$ if and only if $T=T_d$ for some right-continuous $d\in\mathcal F_D$, where $\mathcal F_U^\diamond$ is the set of all strictly increasing and continuous functions $u$ satisfying $u(\mathbb{R})=\mathbb{R}$. To explain the terminology involved, for $F, G\in \M$, we write $F\leq_{\rm st} G$  if $F(x)\geq G(x)$ for all $x\in\mathbb{R}$; $T$ is \emph{monotone} if $T(F)\leq_{\rm st} T(G)$ for $F\leq_{\rm st} G,~ F, G\in \mathcal M$; $T$ is \emph{lower-semicontinuous} if $T(F)\leq_{\rm st} G\in \mathcal M$ whenever $\{F_n\}_{n\in \N}\subseteq \mathcal M$ converges weakly to $F\in \mathcal M$ and satisfies $T(F_n)\leq_{\rm st} G$. In contrast to these properties, our results in Theorem \ref{PTtheorem1} only require one condition to pin down probability distortions. Moreover, Theorem \ref{PTtheorem1} characterizes more general probability distortions than that of \cite{LSW21}, as we do not have lower-semicontinuity.
Below we present an example of a probability distortion $T$ that is not lower-semicontinuous.
 This example shows that Theorem \ref{PTtheorem1} is not only conceptually simpler than the corresponding result of \cite{LSW21}, but also covers more cases.
  \begin{example}\label{ex:2} Define $T(F)=\delta_{F_R^{-1}(1/2)}$, where $\delta_x$ is the point mass distribution with probability $1$ at $x$. Then, $T$ is  monotone and $T\circ T^u=T^u\circ T$ for all $u\in \mathcal F_D$.
  However, $T$ does not satisfy lower semicontinuity. This can be checked as follows. Let $F_n= \mathrm{Bernoulli}(1/2-1/n),~n\geq  2$, $F_0= \mathrm{Bernoulli}(1/2)$ and $G= \mathrm U[0,1]$. Clearly, $F_n$ weakly converges to $F_0$ as $n$ tends to infinity.  By direct calculation, it follows that for $p\in (0,1)$,  $$q_p^T(F_n)=(F_{n})_R^{-1}(1/2)=0<p=G_L^{-1}(p),$$
  where $q_p^T(F)$ is the $p$-th left quantile of $T(F)$, implying $T(F_n)\leq_{\rm st} G$. However,
  $$q_p^T(F_0)=(F_{0})_R^{-1}(1/2)=1>p=G_L^{-1}(p),~p\in(0,1),$$
  implying that $T(F_0)\leq_{\rm st} G$ does not hold.
  Hence $T$ is not lower-semicontinuous.  Applying Theorem \ref{PTtheorem1}, $T$ is still a probability distortion in the sense of (i) of Definition \ref{many properties def}. Let $d(x)=\id_{(1/2,1]}(x),~x\in [0,1]$. Note that $d\in \mathcal F_D$ but $d$ is not right-continuous. Moreover,  $$T_d(F)(x)=d\circ F(x+)=\id_{[ F_R^{-1}(1/2),\infty)}(x)=T(F)(x),~x\in\mathbb{R},$$
  and this verifies that $T=T_d$.
  \end{example}

 The fact that probability distortions can be characterized with only one property may be surprising. It is built on a related result in the recent literature.  Our proof  of Theorem \ref{PTtheorem1}  uses some results of \cite{FLW22}, where quantile functionals are shown to be characterized by only one property, called ordinality, which means commutation with utility transforms for real-valued mappings.

  \section{Utility transforms}\label{Sec:UT}
In a parallel fashion to Section \ref{Sec:PD}, we characterize utility transforms via commutation with probability distortions.  This result establishes that utility transforms are the ``maximal'' class of distributional transforms with respect to which each probability distortion commutes.
\begin{theorem}\label{PTtheorem2}
For a mapping $T: \M\to\M$, $T$ commutes with each element in  $\mathcal{D}$ if and only if $T\in\mathcal{U}$.
  \end{theorem}

Theorem \ref{PTtheorem2} establishes that for a distributional transform $T$, if distorting the input distribution is the same as distorting the output distribution with the same probability distortion, then the distributional transform is a probability distortion.

Although being nicely parallel with Theorem \ref{PTtheorem1}, the idea of pinning down utility transforms via probability distortions in Theorem \ref{PTtheorem2} is novel.
We are not aware of any results in the literature that characterizing utility transforms.
Mathematically, the proof of Theorem \ref{PTtheorem2} does not follow directly from that of  Theorem \ref{PTtheorem1}, because the roles of probability distortions (changing the distribution function) and utility transforms (changing the quantile functon) are not exactly symmetric.

%We can conclude from  that probability distortions and utility transforms can be mutually determined via commutation  with each other.
%\begin{corollary}\label{cor:important} For a mapping  $T: \M\to\M$,
%  \begin{enumerate}
%  \item[(i)] $T$ commutes with each element of $\mathcal{U}$ if and only if $T\in\mathcal F_D$;
%  \item[(ii)] $T$ commutes with each element of $\mathcal{D}$ if and only if $T\in\mathcal U$.
%  \end{enumerate}
 % The same statements hold true if $\mathrm{U}$ is replaced by $\mathrm{U}^L$ and $\mathrm{D}$ by $\mathrm{D}^R$.
%\end{corollary}

Theorems \ref{PTtheorem1} and \ref{PTtheorem2} offer a mathematical duality between the expected utility theory and the dual utility theory of \cite{Y87}, in the sense that one can be derived from the other.
Taking the the class of increasing and continuous utility functions as given, if we look at the expected value of any transform that commutes with these utility transforms, then we arrive at a dual utility.
Conversely,
taking the the class of all distortion functions as given, if we look at the expected value of any transform that commutes with these probability distortions, then we arrive at an expected utility.
This connection is new to the literature.

The separation of  marginal utility (modelled by UF) and the probabilistic risk (modelled by DF) attitudes under RDU (see Definition \ref{def} in Section \ref{Sec:RDU}) are considered by \cite{W94}, where these attitudes are characterized independently of each other to allow for comparison of risk attitudes across decision makers.   This perspective is different from Theorems \ref{PTtheorem1}-\ref{PTtheorem2}, where utility transforms and probability distortions are mutually determined by commutation with each other.

%\begin{corollary} For a mapping  $T: \M\to\M$,
%  \begin{enumerate}
%  \item[(i)]  $T$ commutes with $T^u$ for all $u\in \mathcal U^L$ if and only if $T=T_d$ for some $d\in \mathcal F_D^R$;
%  \item[(ii)]  $T$ commutes with $T_d$ for all $d\in \mathcal F_D^R$ if and only if $T=T^u$ for some $u\in \mathcal U^L$.
%  \end{enumerate}
%\end{corollary}

\section{RDU transforms}\label{Sec:RDU}
%Let $\mathcal F_U$ denote the set of all strictly increasing and continuous functions $u$ such that $u(\R)=\R$, and

In the two sections above, we have characterized both the class of utility transforms
and the class of probability distortions. A natural question is whether an RDU transform, which is a composition of a utility transform and a probability distortion, also admits a characterization in a similar fashion. We address this question in this section.

We first formally define an RDU transform.
%Let $\mathcal F_U$ denote the set of all strictly monotone and continuous functions $u$ such that $u(\R)=\R$. The reason to consider decreasing functions in $\mathcal F_U$ is because decreasing utility transforms have similar commutation properties to increasing ones.
\begin{definition}\label{def}
A distributional transform $T:\M\to\M$ is an \emph{RDU transform} if there exist $d\in\mathcal F_D$ and $u\in\mathcal F_U$ such that $T=T_d\circ T^u$. Here, $u$ is called the UF of $T$ and $d$ is called the DF of $T$.
\end{definition}

 To understand special properties of RDU transforms, we first observe that an RDU transform does not commute with  utility transforms  in general, as implied by Theorem \ref{PTtheorem1}.
 Essentially, this is because utility transforms do not commute with each other; for instance, $2(x+1)$ is not equal to $2x+1$.
 The same non-commutation of RDU  transforms holds also with probability distortions.
 Therefore, in order to characterize RDU transforms,
 we need to seek for weaker properties than commutation with utility transforms (or probability distortions).
Although a utility transform $T^{u_1}$ with $u_1\in \mathcal F_U$ does not commute with another one $T^{u_2}$,
 assuming that $u_1$ is strictly monotone and surjective (i.e.,  $u_1(\R)=\R$),
 there exists $u_3\in \mathcal F_U$ such that
\begin{align}
\label{eq:commute-u}
T^{u_3}\circ T^{u_1}= T^{u_1} \circ T^{u_2}; \mbox{~equivalently,~} u_3\circ u_1= u_1 \circ u_2,
\end{align}
and such $u_3$ is  given by $u_3 = u_1 \circ u_2 \circ u_1^{-1}$.
  This inspires to use the property of set commutation.
  For a set $\mathcal T$ of functions and a function $T$,
  we write $\mathcal T\circ T = \{T'\circ T: T'\in \mathcal T\}$,
  and similarly,
 $ T\circ \mathcal T = \{T\circ T': T'\in \mathcal T\}$.

\begin{definition}  A distributional transform
  $T:\M\to\M$ is said to \emph{commute} with a set $\mathcal T$ of distributional transforms on $\mathcal M$
  if  $\mathcal T\circ T=T\circ \mathcal T$.
\end{definition}

Stated equivalently,  $\mathcal T\circ T=T\circ \mathcal T$ means that
for any $T_R\in\mathcal T$ there exists $T_L\in \mathcal T$ such that $T_L\circ T=T\circ T_R$,
 and for any $T_L\in \mathcal T$ there exists $T_R\in \mathcal T$ such that $T_L\circ T=T\circ T_R$.
Clearly,  commutation with  the set $\mathcal T$ is weaker than commutation with each element of $\mathcal T$, which further requires   $T_L=T_R$ in  both statements.

%Let $\mathcal F_U=\{T^u:u \in \mathcal F_U\}$.
\begin{theorem}\label{RDUM}
Let  $T:\M\to\M$  be a monotone distributional transform. Then
 $T $ commutes with $\mathcal{U}$ if and only if
    $T=T_d\circ T^u$, where $u$ is strictly increasing and surjective.
 \end{theorem}

 The reason why the UF of $T$ in Theorem \ref{RDUM} needs to be strictly increasing and surjective
can be seen from the above discussion; we can check that both properties of $u_1$ are needed for the existence of $u_3 $ satisfying \eqref{eq:commute-u} for each $u_2 $ and the existence of $u_2$ satisfying \eqref{eq:commute-u}  for each $u_3$.

Set commutation with $\mathcal T$ has a similar interpretation to commutation in Section \ref{Sec:UT} as the effect of rescaling the input and the output distributions, but now the  output rescaling can be different from the input one; yet these two scaling operations belong to the same class $\mathcal{U}$. If we remove monotonicity of $T$ in Theorem \ref{RDUM}, we will include mappings like $T=T_d\circ T^u$ with strictly decreasing and surjective $u$ in the conclusion of Theorem \ref{RDUM}, which is undesirable in decision theory.

The duality between utility transforms and probability distortions
  hints at a similar result to Theorem \ref{RDUM} formulated with probability distortions.
This intuition checks out but it requires some technical work to formalize.  For this, we need to consider a set smaller than $\mathcal D$.  Let $\mathcal F_U^L$ be the set of all increasing and  left-continuous functions on $\R$ and let $\mathcal F_D^R$ be the set of all right-continuous functions  in $\mathcal F_D$. Note that  $\mathcal F_U\subseteq\mathcal F_U^L$ and $\mathcal F_D^R\subseteq\mathcal F_D$. The utility transform can be naturally extended to $\mathcal F_U^L$ and the RDU can also be extended to the case with increasing and left-continuous UF. We denote $\mathcal{U}^L=\{T^{u}:u \in \mathcal F_U^L\}$ and $\mathcal D^R=\{T_d: d\in \mathcal F_D^R\}$.

\begin{theorem}\label{RDUM2}
Let  $T:\M\to\M$  be a monotone distributional transform. Then
     $T $ commutes with $\mathcal{D}^R$ if and only if
    $T$ is an RDU transform with left-continuous UF and strictly increasing and continuous DF.
 \end{theorem}

The interpretation of Theorem \ref{RDUM2} is similar to that of Theorem \ref{RDUM}. Roughly, RDU transforms are precisely those which are ``exchangeable" (in the sense of $T\circ \mathcal T=\mathcal T\circ T$) with either scaling or distortion as a whole class.

%Theorem \ref{RDUM} offers a very simple characterization on RDU, but we are unable to find an interpretation in Theorem \ref{RDUM} economically. %states that for a distributional transform $T$, if utilizing the input distribution with one utility function is the same as utilizing the output distribution with another utility function and meanwhile utilizing the output distribution with one utility function is the same as utilizing the input distribution with another utility function, then the distributional transform is an RDU transform.

%\begin{remark}
%With slight adjustment, an RDU transform can be characterized by set commutation with the set of probability distortions.
% We omit this formulation due to its similarity to weak commutation with utility transforms. \com{maybe explain this in the appendix.}
% \com{See whether this is ok.}
%\end{remark}

\section{Discussions}
\label{sec:6}
Our results offer a fundamental symmetry between utility transforms and probability distortions.
In spirit,
  this is similar to a duality between the  expected utility  (EU) theory of \cite{vNM47} and the dual utility (DU) theory  of \cite{Y87}.
  Nevertheless, the duality between EU and DU is based on parts of the representation, rather than behavior.
  %conceptual but not mathematical one.
EU assumes the independence axiom (among others) to arrive at the preference representation $F\mapsto \int u(x)\d F(x)$ for some utility function $u \in \mathcal F_U$.
DU assumes a dual-independence axiom (plus the same other axioms)
to arrive at the preference representation $F\mapsto \int x \d (d\circ F)(x)$ for some distortion function $d\in \mathcal F_D$.
Although conceptually plausible, it is not clear  how one could start from the independence axiom to derive the dual independence axiom, and vice versa.
Our results in Theorems \ref{PTtheorem1} and \ref{PTtheorem2} gives a way to obtain from utility transforms to probability distortions, and also the other way around.
Since EU and DU are represented by the expected value of the transformed risk distribution,
our results also suggest that EU and DU can be mathematically derived from each other.

 Regarding RDU transforms, there are also some interesting observations. %Let $\mathcal F_D$ be the set of  continuous and strictly increasing functions in $\mathcal F_D$ and $\mathcal F_U^*$ be the set of strictly increasing functions in $\mathcal F_U$.
 Let $\mathcal R$ be the class of RDU transforms of the form $T_d\circ T^{u}$
 where $d\in \mathcal F_D$ and $u\in \mathcal F_U$.
 First, note that  the order of the composition of  $T_d $ and $T^{u}$ for an RDU transform does not matter, justified by either Theorem \ref{PTtheorem1} or \ref{PTtheorem2}.
 Second, due to the same reason, the class  $\mathcal R$ is closed under composition,
 and so are $\mathcal{U}$ and $\mathcal D$.
 Therefore, $\mathcal R$ equipped with composition is a semigroup with $\mathcal U$ and $\mathcal D$ being its sub-semigroups. Moreover, $\mathcal R$ is the smallest semigroup containing both $\mathcal U$ and $\mathcal D$.
 This, from an algebraic perspective, may justify that RDU is the most natural generalization of EU and DU compared to all other ways, such as a linear combination or a maximum of EU and DU (which are obviously strange).

 Next, we rephrase our results in the context of decision making.
Assume that a decision maker has a unique \emph{perception}, which describes  how she perceives  a distribution in $\mathcal M$  presented to her, and  her preference is   represented by the mean of her perception.
Certainly, the perception is precisely a distributional transform on $\mathcal M$.

We assume that the decision maker is described by her perception, instead of a preference relation as in the classic literature.
%
%In the following discussions,
% the decision maker is mathematically identified with the perception.
%This is   different from identifying the decision maker with a preference represented by the mean of her perception.
%To see the difference,
%the EU preference with utility function $g$ is represented by the mean of the perception $T^{[g]}$
%but also represented by the mean of the other perception $T: F\mapsto \delta_{x}$ where $\delta_x$ is the point-mass at $x=\int g (y)\d F(y)$, although the second way of representation is awkward to interpret.
%Thus, for a given preference,
%  the underlying perception may not be unique.
%  We focus on perceptions below.
We say that a decision maker is an EU thinker if her perception is a utility transform  in $\mathcal{U}$ (thus, the preference is EU); she is a DU thinker if her perception is  a probability distortion in $\mathcal D$ (thus, the preference is DU); she is an RDU thinker if her perception is  the composition of a utility transform in $\mathcal{U}$ with a strictly increasing and surjective UT and a probability distortion in $\mathcal D$ (thus, the preference is RDU in Theorem \ref{RDUM}).
The next proposition is immediate from Theorems \ref{PTtheorem1}-\ref{RDUM}.
\begin{proposition}\label{prop:discussion}
For a decision maker equipped with a perception,
\begin{enumerate}[(i)]
\item she is   a  DU thinker   if and only if
her perception commutes with  each EU perception;
\item she is  an EU thinker if and only if
her perception commutes with each DU perception;
\item she is   an RDU thinker  if and only if
her perception commutes with the set of EU perceptions.
\end{enumerate}
\end{proposition}

A decision maker with perception $T$ perceives a distribution $F$ as $G= T(F)$.
Commuting with EU perceptions mathematically means that she perceives $F\circ u^{-1}$ as $G\circ u^{-1}$ for all $u\in \mathcal F_U$. This means intuitively that her perception is independent of the metric or scale (e.g., linear, log, exponential, etc) used for the underlying random outcome.
On the other hand, commuting with DU perceptions mathematically means that she perceives $d \circ F  $ as $d\circ G $ for all $d\in \mathcal F_D$. In other words, her perception is independent of a probabilistic weighting.
These explanations are actually quite obvious:  a DU thinker applies a distortion $d$  to a (cumulative) probability $p$ regardless of its outcome value, and an EU thinker applies a utility $u$ to an outcome $x$ regardless of its probability.
Putting Proposition \ref{prop:discussion} into the above context, the independence principle above alone is able to characterize  both DU and EU thinkers, and a slight variation is able to characterize RDU thinkers.

As for future research, let $\mathcal{T}_0$ be the set of all  distributional transforms, i.e., mappings from $\M$ to $\M$. We may define $\mathcal{C}:2^{\mathcal T_0} \to 2^{\mathcal T_0} $ by $\mathcal{C}(\mathcal{T})=\bigcap_{T'\in \mathcal{T}}\{T\in \mathcal T_0:T\circ T'=T'\circ T\}$, that is,  the set of elements of $\mathcal{T}_0$ that commute with respect to each element of $\mathcal{T}$.  It is easy to see that $\mathcal{C}$ forms an (antitone) Galois connection with itself with respect to set inclusion on $2^{\mathcal{T}_0}$ (see e.g., \citet{Bly05}).  This implies that $\mathcal{C}^2:=\mathcal{C}\circ\mathcal{C}$ forms a closure operator, and the results in Theorems~\ref{PTtheorem1} and \ref{PTtheorem2} indicate that the sets $\mathcal{U}$ and $\mathcal{D}$ are closed sets according to this operator, with the further property that $\mathcal{U}=\mathcal{C}({\mathcal D})=\mathcal{C}^2({\mathcal U}) $ and $\mathcal{D}=\mathcal{C}(\mathcal{U})=\mathcal{C}^2({\mathcal D})$.  This latter property is a formal expression of the term ``duality''.  A general investigation of the map $\mathcal{C}$, together with its implied lattice of closed sets, could be fruitful.  
%Summarizing them in the following conjecture, we have a nicely interpretable   statement, which would follow directly from Conjecture \ref{conj:1}.
%
%\begin{conjecture}\label{conj:2}
%A decision maker is a DT decision maker if and only if
%her perception  is independent of the outcome scale, and
%a decision maker is an EU decision maker if and only if
%her perception is  independent of the probability scale.
%\end{conjecture}
%
%Conjecture \ref{conj:2} (if true) gives a new characterization of EU and DT.

%\subsection*{Acknowledgements}
%RW acknowledges financial support from the Natural Sciences and Engineering Research Council of Canada (RGPIN-2018-03823 and CRC-2022-00141).

\appendix

\appendix
\section{Proof of Theorem \ref{PTtheorem1}}\label{Sec:AA}
For $p\in (0,1]$, let $q_p:\M \to \R$ be the left quantile at probability level $p$, that is, $q_p(F)=F^{-1}_L(p)$.
Further, for $T: \M\to\M$, define $q^T_p: \M\to \R$ by $q^T_p(F)= q_p(T(F))$. For an increasing function $u$, let $u_R^{-1}(x)=\sup\{y\in\mathbb{R}: u(y)\leq x\}$ with $\sup\emptyset=-\infty$. For notational convenience, we extend the domain of $F\in\M$ to $[-\infty,\infty]$ by letting $F(-\infty)=0$ and $F(\infty)=1$.
%We characterize   probability distortions on $\M$.

\begin{proof}[Proof of Theorem \ref{PTtheorem1}]
  We first prove the   {``if"} part; that is,  we will show  that $T_d$ and $T^u$ commute.  For $d\in\mathcal F_D$, let
   $
T(F)(x)=T_d(F)(x)=(d\circ F)(x+),~x\in\mathbb{R}$. Note that for any $F\in\M$ and $u\in\mathcal F_U$,
 $$T_d\circ T^u(F)(x)=[d\circ(F\circ u_R^{-1})](x+)=\lim_{y\downarrow x}d(F(u_R^{-1}(y))),$$
 and $$T^u\circ T_d (F)(x)=w(u_R^{-1}(x)), $$
 where $w(x)=(d \circ F)(x+)=\lim_{y\downarrow x}d(F(y))$ for $x\in\mathbb R$ and $w(+\infty)=1, w(-\infty)=0.$
 Moreover, as $u\in\mathcal F_U$,  if $u_R^{-1}(x)\in\mathbb{R}$, then $u_R^{-1}(y)>u_R^{-1}(x)$ for $y>x$. Hence,
 $\lim_{y\downarrow x}d(F(u_R^{-1}(y)))=w(u_R^{-1}(x))$, implying $T_d\circ T^u(F)(x)=T^u\circ T_d (F)(x).$
 For $u_R^{-1}(x)=\infty$, by definition, we have
 $T_d\circ T^u(F)(x)=T^u\circ T_d (F)(x)=1.$
 For $u_R^{-1}(x)=-\infty$, we have $\lim_{y\downarrow x}u_R^{-1}(y)=-\infty$, which implies
 $T_d\circ T^u(F)(x)=T^u\circ T_d (F)(x)=d(0)=0.$
 Consequently, $T_d\circ T^u=T^u\circ T_d$.

 Next, we show the  {``only if"} part. For each $p\in (0,1)$, we claim that the functional $q^T_p: \M \to \R$ satisfies
$q^T_p(T^u(F)) = u(q^T_p(F))$ for all $u\in\mathcal F_U$.
%\item   for all sequences $\{F_n\}_{n\in \N}\subseteq \mathcal M$ which weakly converges to $F\in \mathcal M$ satisfying
 %$q^T_p(F_n)\leq x\in \R$, we have  $q^T_p(F)\leq x$.
 Using the commutation of $T$ and $T^u$ for $u\in\mathcal F_U$ and the fact that $q_p$ commutes with $T^u$ for $u\in\mathcal F_U$, we have
  $$
 q^T_p(T^u(F)) = q_p (T \circ T^u (F)) = q_p ( T^u\circ T   (F)) = u(q_p( T   (F))) = u(q^T_p(F)).
 $$
 By Theorem 2 of \cite{FLW22},  $q^T_p$ is a (left or right) quantile of $F$  at a fixed level for all $F\in\M$.  We denote this level by $g(p)\in [0,1]$. % If $d(p)=1,$ then $q_p^T(F)=F_L^{-1}(d(p))$; if $d(p)=0,$ then $q_p^T(F)=F_R^{-1}(d(p))$; if   $d(p)\in (0,1)$, then
% Hence, $$q_p^{T}(F)=F_L^{-1}(d(p))~\text{or}~ q_p^{T}(F)=F_R^{-1}(d(p)).$$
 Hence there exists $E\subseteq (0,1)$ such that for any $F\in\M$
 \begin{equation}\label{eq:quantile}q_p^T(F)=\left\{\begin{array}{cc}
   F_L^{-1}(g(p)),& ~p\in E\\
   F_R^{-1}(g(p)),&~p\in (0,1)\setminus E
   \end{array}\right..
   \end{equation}
 %Let $E_1=\{d(p): p\in (0,1)\}$. If $1\in E_1$, then $1\in E$; If $1\notin E_1$, we can assume that $1\in E$, which does not affect (\ref{eq:quantile}). Analogously, we can assume that $0\notin E$.  Now  $1\in E$ and $0\notin E$}.
   Note that by definition, $p\notin E$ if $g(p)=0$ and $p\in E$ if $g(p)=1$ . Clearly,   $g(p)$ is an increasing function over $(0,1)$.
By (\ref{eq:quantile}), we have
\begin{align}\label{TF}
T(F)(x)&=\lambda\left(\{p\in E: F_L^{-1}(g(p))\leq x\}\cup\{p\in (0,1)\setminus E: F_R^{-1}(g(p))\leq x\}\right)\nonumber\\
&=\lambda\left(\{p\in E: g(p)\leq F(x)\}\cup\{p\in (0,1)\setminus E: F_R^{-1}(g(p))\leq x\}\right), ~x\in\mathbb{R},
\end{align}
where $\lambda$ is the Lebesgue measure on $[0,1]$. Note that for $F(x)\in (0,1)$
\begin{align*}
\{p\in (0,1)\setminus E: F_R^{-1}(g(p))\leq x\}=\left\{\begin{array}{cc}
\{p\in (0,1)\setminus E: g(p)< F(x)\},& x<F_R^{-1}(F(x))\\
\{p\in (0,1)\setminus E: g(p)\leq F(x)\},& x=F_R^{-1}(F(x)).
\end{array}\right.
\end{align*}
 Hence
\begin{align}\label{eq:g}
d(F(x)-)\leq T(F)(x)\leq d(F(x)),~x\in (F_R^{-1}(0), F_L^{-1}(1)),
\end{align}
where $d(x)=\lambda(\{p\in (0,1): g(p)\leq x\})$ is a right-continuous and increasing function with  $d(1)=1$. Note that $d(0)\in [0,1]$.
Moreover, by (\ref{TF}) we have  \begin{align}\label{TF1}T(F)(x)=\left\{\begin{array}{cc}
0,& x<F_R^{-1}(0)\\
d(1),& x\geq F_L^{-1}(1),
\end{array}\right.
\end{align}
and
 for  $x=F_R^{-1}(0)$,
$T(F)(x)=d(0)$ if $F(F_R^{-1}(0))=0$;
$d(F(x)-)\leq T(F)(x)\leq d(F(x))$ if $F(F_R^{-1}(0))>0$.
Let $D$ be the set of all discontinuous points of $d$ on $(0,1)$. For $x\in D$, let
$$r(x):=\lambda\left(\{p\in E: g(p)\leq x\}\cup \{p\in (0,1)\setminus E: g(p)<x\}\right).$$
Noting that $d(x-)\leq r(x)\leq d(x),~x\in D$, it follows that the function $\widehat d$ given by
\begin{align*}
\widehat{d}(x)=\left\{\begin{array}{cc}
d(x),& x\in (0,1]\setminus D\\
r(x),& x\in D\\
0,& x=0
\end{array}
\right.
\end{align*}
is an increasing function on $[0,1]$ with $\widehat{d}(0)=0$ and $\widehat{d}(1)=1$.
It follows from (\ref{eq:g}) that
\begin{align*}
T(F)(x)=\widehat{d}(F(x)),~F(x)\in (0,1)\setminus D.
\end{align*}
 Moreover, by (\ref{TF1}), we have
$$ T(F)(x)=\widehat{d}(F(x)),~x\in (-\infty, F_R^{-1}(0))\cup [F_L^{-1}(1),\infty).$$
For  $F(x)=c\in D$, if $x<F_R^{-1}(F(x))$,
\begin{align*}
T(F)(x)=\lambda\left(\{p\in E: g(p)\leq c\}\cup \{p\in (0,1)\setminus E: g(p)<c\}\right)=r(c)=\widehat{d}(F(x)).
\end{align*}
For  $F(x)=c\in D$ and $x=F_R^{-1}(F(x))$, we have $T(F)(x)=d(c)$, which may  not be equal to $\widehat{d}(F(x))$.
Hence,
  $
  T(F)(x)\neq \widehat{d}(F(x))
$ holds  over $(F_R^{-1}(0),F_L^{-1}(1))$ only if $F(x)\in D$ and $x=F_R^{-1}(F(x))$.
This implies that $T(F)(x)=\widehat{d}(F(x))$ does not hold over $(F_R^{-1}(0),F_L^{-1}(1))$ only at countable number of points. Hence, by right continuity of the two functions, we have
$$T(F)(x)=(\widehat{d}\circ F)(x+),~x\in (F_R^{-1}(0),F_L^{-1}(1)).$$
Moreover, by   right continuity of $T(F)(x)$ and $(\widehat{d}\circ F)(x+)$, we have
$$T(F)(F_R^{-1}(0))=(\widehat{d}\circ F)(x+)|_{x=F_R^{-1}(0)}.$$
Combing the above results, we conclude   $T(F)(x)=(\widehat{d}\circ F)(x+)$ for all $x\in \mathbb{R}.$
This completes the proof of Theorem \ref{PTtheorem1}.
\end{proof}
  \section{Proof of Theorem \ref{PTtheorem2}}

\begin{proof}[Proof of Theorem \ref{PTtheorem2}] The ``if" part is the same as Theorem \ref{PTtheorem1}.   We  focus on the ``only if" part.
 For $n\geq 1$,  denote
$\mathcal{C}_n=\{F\in \M: F(-n)=0,~F(n)=1\},$
and let $F_{U_n}$ represent the uniform distribution on $[-n,n]$. Clearly,
$$\mathcal{C}_n=\{T_d(F_{U_n}): d~\text{is right-continuous and }~d\in \mathcal F_D\},~\mathcal{C}_n\subseteq\mathcal{C}_{n+1}, n\geq 1, ~\text{and}~\M=\bigcup_{n=1}^\infty\mathcal{C}_n.$$
 By $T\circ T_d=T_d\circ T$, it follows that for $p\in (0,1)$ and $F=T_d(F_{U_n})$ with  a right-continuous $d\in \mathcal F_D$
$$q_p^T(F)=q_p^T(T_d(F_{U_n}))=q_p^{T_d}(T(F_{U_n}))=q_{d_L^{-1}(p)}^T(F_{U_n}).$$
Note that $d_L^{-1}(p)=( {F_L^{-1}(p)+n})/({2n}),~p\in (0,1)$ and $d_L^{-1}(p)>0$ for $p\in (0,1)$. Define $u_n:(-n,n]\to\mathbb{R}$ by $u_n(t)=q_{(t+n)/2n}^T(F_{U_n})$.
Then, for $F\in \mathcal{C}_n$,
$$q_p^T(F)=u_n(F_L^{-1}(p)),~p\in (0,1).$$
Setting $F=F_{U_n}$, it follows that for $n\geq 1$ and $m\geq 0$,
$$q_p^T(F_{U_n})=u_n(-n+2np)=u_{n+m}(-n+2np),~p\in (0,1).$$
Hence, for $n\geq 1$ and $m\geq 0$, we have $u_{n+m}(t)=u_n(t),~t\in (-n,n)$. Moreover, one can verify   $u_{n+m}(n)=u_n(n)$.  Therefore, we can define a function $u:\mathbb{R}\to\mathbb{R}$ by letting $u(t)=u_n(t),~t\in (-n,n]$. Clearly, $u$ is an increasing and left-continuous function. Moreover, we have for $F\in\M$
\begin{align}\label{Eq:utility}q_p^T(F)=u(F_L^{-1}(p)),~p\in (0,1). \end{align}
This implies that for any $F\in\M$, $T(F)=T^{u}(F)$.

We next  show   continuity of $u$. We assume by contradiction that there exists $x_0\in\mathbb{R}$ such that $u(x_0+)>u(x_0)$. Denote $y_0=u(x_0)$ and without loss of generality, we assume $x_0\in (0,1)$. Let $d(x)=\id_{(x_0,1]}(x),~x\in [0,1]$ and $F_U$ be the uniform distribution on $[0,1]$.  Note that $d\in\mathcal F_D$ and
$(d \circ F_U)(x+)=\id_{[x_0,\infty)}(x)$. Hence
$$T^{u}\circ T_d(F_U)(y_0)=\id_{[x_0,\infty)}(u_R^{-1}(y_0))=\id_{[x_0,\infty)}(x_0)=1,$$
and
$$T_d\circ T^{u}(F_U)(y_0)=\lim_{y\downarrow y_0}d(F_U(u_R^{-1}(y)))=d(F_U(x_0))=d(x_0)=0.$$
This implies that $T^{u}\circ T_d\neq T^{u}\circ T_d$ for some $d\in \mathcal F_D$, yielding a contradiction.
Hence $u$ is continuous.
 This completes the proof.
\end{proof}
\section{Proof of Theorem \ref{RDUM}}\label{Sec:AC}
In order to prove Theorem \ref{RDUM}, we need  an additional result for mappings $\rho:\M\to\R$.   We say that $\rho:\M\to\R$ commutes with the set $\mathcal U$ if $\mathcal F_U\circ \rho=\rho\circ\mathcal U$, where $\mathcal F_U\circ \rho=\{u\circ \rho: u\in\mathcal F_U\}$ and $\rho\circ\mathcal U=\{\rho\circ T: T\in \mathcal U\}$.   For a mapping $\rho:\M\to\R$, we say $\rho$ is  a \emph{left quantile} if there exists some $p\in (0,1]$ such that $\rho(F)=F_L^{-1}(p)$ for all $F\in\M$; $\rho$ is a \emph{right quantile} if there exists some $p\in [0,1)$ such that $\rho(F)=F_R^{-1}(p)$ for all $F\in\M$; $\rho$ is a \emph{quantile} if $\rho$ is either a left or  right quantile.
 \begin{proposition}\label{mappings} For  a mapping $\rho:\M\to\R$,  $\rho$ commutes with $\mathcal U$ if and only if there exists a strictly monotone, continuous and surjective $h$ such that $\rho=h\circ \widehat{\rho}$, where $\widehat{\rho}:\M\to\R$ is a quantile.
 \end{proposition}
 \begin{proof} We first consider the ``if" part. For $u\in\mathcal F_U$, let $\psi=h\circ u\circ h^{-1}$. Note that $\psi\in\mathcal F_U$. Using the commutation of utility transforms and quantiles (Theorem 2 of \cite{FLW22}), we have $$\psi\circ\rho=\psi\circ h\circ \widehat{\rho}=h\circ u\circ \widehat{\rho}=h\circ\widehat{\rho}\circ T^u=\rho\circ T^u.$$ Let $\omega=h^{-1}\circ u\circ h$. Clearly, $\omega\in\mathcal F_U$. Analogously as above, we have
$$\rho\circ T^{\omega}=h\circ \widehat{\rho}\circ T^{\omega}=h\circ \omega\circ \widehat{\rho}=u\circ h\circ \widehat{\rho}=u\circ\rho.$$
Note that in the above equations, we only use the commutation of $\widehat{\rho}$ and increasing transforms.
Hence $\rho$ commutes with $\mathcal U$.

  We next show the ``only if" part. By $\mathcal U$-commutation, we have for any $u\in\mathcal F_U$, there exists $\psi, \omega\in\mathcal F_U$ such that $\psi\circ \rho(\delta_x)=\rho\circ T^u(\delta_x)$ and $\rho\circ T^{\omega}(\delta_x)=u\circ\rho(\delta_x)$. We denote $h(x)=\rho(\delta_x),~x\in\R$. Then direct computation yields $\psi\circ h=h\circ u$ and $h\circ\omega=u\circ h$. This implies  $h$ is continuous and strictly monotone and $h(\R)=\R$. The proof is given below.

 We assume by contradiction that $h(\R)$ has a lower or upper bound, i.e.,  $h(\R)\subseteq [n_1,\infty]$ for some $n_1\in \R$, or $h(\R)\subseteq (-\infty,n_2]$  for some $n_2\in \R$.  Without loss of generality, we consider the case $h(\R)\subseteq (-\infty,n_2]$ for $n_2\in\R$. Let $u\in\mathcal F_U$  be such that $u(h(0))>n_2+1$. Note that for any $\omega\in\mathcal F_U$, $h\circ\omega(0)\leq n_2<u\circ h(0)$. This means that $h\circ\omega=u\circ h$ does not hold for all $\omega\in\mathcal F_U$, leading to a contradiction. Hence $h(\R)$ is neither bounded from below nor from above. If $h$ is not monotone,  then there exist $a, b, c\in\R$  with $a<b<c$ such that $h(a)<h(b)$ and $h(b)>h(c)$, or  $h(a)>h(b)$ and $h(b)<h(c)$. For the former case, let $u\in\mathcal F_U$ be such that $u(a)=b$ and $u(b)=c$. It follows that $h\circ u(a)>h\circ u(b)$. Moreover,  by the fact that there exists $\psi\in\mathcal F_U$ such that   $\psi\circ h=h\circ u$, we have $\psi\circ h(a)>\psi\circ h(b)$, which contradicts with $h(a)<h(b)$ and $\psi\in\mathcal F_U$. The latter case is not possible by the same reasoning. Hence,  $h$ is a monotone function.  Next, we assume $h$ is flat over $[a,b]$ with $a,b\in\R$ and $a<b$. As $h(\R)$ is not a singleton, there exist $c,d\in\R$ with $c<d$ such that $h(c)\neq h(d)$. Let $u\in\mathcal F_U$ be such that $u(a)=c$ and $u(b)=d$. Hence we have
$h\circ u(a)\neq h\circ u(b)$. In contrast, $\psi\circ h(a)=\psi\circ h(b)$ for all $\psi\in\mathcal F_U$, leading to a contradiction.  Thus, $h$ is strictly monotone. Next, assume by contradiction that $h$ is not continuous on $\R$. Then there exist $x_0, x_1\in\R$ such that $h$ is not continuous at $x_0$ but  continuous at $x_1$. Let $u:x\mapsto x-x_1+x_0,~x\in\R$. It follows that $u\in\mathcal F_U$ and $h\circ u$ is not continuous at $x_1$. However, $\psi\circ h$ is continuous at $x_1$ for all $\psi\in\mathcal F_U$, which contradicts with the fact that there exists  $\psi\in\mathcal F_U$ such that  $\psi\circ h=h\circ u$. Hence $h$ is continuous on $\R$. Combining all the properties proved above, we conclude that $h$ is  continuous, strictly monotone and satisfying  $h(\R)=\R$.

Using the above conclusion, $\psi\circ h=h\circ u$ implies $\psi=h\circ u\circ h^{-1}$. Hence, $$\psi\circ\rho=h\circ u\circ h^{-1}\circ\rho=\rho\circ T^u,$$ which further implies $u\circ h^{-1}\circ\rho=h^{-1}\circ\rho\circ T^u$. Let $\widehat{\rho}=h^{-1}\circ\rho$. It follows that $\widehat{\rho}:\M\to\R$ and $u\circ \widehat{\rho}=\widehat{\rho}\circ T^u$ for all $u\in\mathcal F_U$. Using Theorem 2 of \cite{FLW22}, we obtain that $\widehat{\rho}$ is a quantile. We further have $\rho=h\circ \widehat{\rho}$.
 \end{proof}

The result in Proposition \ref{mappings} leads to the proof of Theorem \ref{RDUM}.

\begin{proof}[Proof of Theorem \ref{RDUM}]
The ``if" part follows by  checking the definition. Note that $T=T_d\circ T^{h}$ for $d\in\mathcal F_D $ and $h\in\mathcal F_U^\diamond$. For $u\in\mathcal F_U$, let $\psi=h\circ u\circ h^{-1}\in\mathcal F_U$.
It follows from  Theorem \ref{PTtheorem1}  that
$$T^{\psi}\circ T=T^{\psi}\circ T_d\circ T^{h}=T_d\circ T^{\psi}\circ T^{h}=T_d\circ T^{\psi\circ h}=T_d\circ T^{h\circ u}=T\circ T^u.$$
Moreover, let $\omega=h^{-1}\circ u\circ h\in\mathcal F_U$. Analogously, we have
$$T\circ T^{\omega}=T_d\circ T^{h}\circ T^{\omega}=T_d\circ T^{h\circ \omega}=T_d\circ T^{u\circ h}=T^u\circ (T_d\circ T^{h})= T^u\circ T.$$
Hence $T$ commutes with $\mathcal U$.

We next focus on the ``only if" part. Recall that for $p\in (0,1)$, $q^T_p: \M\to \R$ is given by $q^T_p(F)=(T(F))_L^{-1}(p)$. Note that for any $u\in\mathcal F_U$, there exist $\psi, \omega\in\mathcal F_U$ such that $T^{\psi}\circ T=T\circ T^u$ and $T\circ T^{\omega}=T^u\circ T$. This implies that
$\psi\circ q^T_p=q^T_p\circ T^u$ and $q^T_p\circ T^{\omega}=u\circ q^T_p$ for all $p\in (0,1)$. Hence,
 $q^T_p$ commutes with $\mathcal U$ for all $p\in (0,1)$. By Proposition \ref{mappings} and  monotonicity of $T$, for $p\in (0,1)$, there exists $h_p\in \mathcal F_U^\diamond$ such that $q^T_p=h_p\circ \widehat{\rho}_p$, where $\widehat{\rho}_p$ is a quantile. Note that $\psi\circ q^T_p(\delta_x)=q^T_p\circ T^u(\delta_x)$ for all $x\in\R$ and $p\in (0,1)$. Hence we have
 $\psi\circ h_p=h_p\circ u$ for all $p\in (0,1)$. A simple manipulation yields that $h_{p_1}\circ u\circ h_{p_1}^{-1}=h_{p_2}\circ u\circ h_{p_2}^{-1}$ holds for all $u\in\mathcal F_U$ and $p_1, p_2\in (0,1)$, which is equivalent to $h_{p_2}^{-1}\circ h_{p_1}\circ u=u\circ h_{p_2}^{-1}\circ h_{p_1}$ for all $u\in\mathcal F_U$ and $p_1, p_2\in (0,1)$.  Denote   $f=h_{p_2}^{-1}\circ h_{p_1}$.
 Then $f\in\mathcal F_U^\diamond$ and $f\circ u=u\circ f$ for all $u\in\mathcal F_U$. Let $u(x)=\min(x, b),~x\in\R$ for some $b\in\R$. It follows that for $x\geq b$, $f\circ u(x)=f(b)$ and $u\circ f(x)=\min(f(x), b)$. Hence $f(b)=\min(f(x), b)$ for $x\geq b$. Letting $x\to\infty$, we have $f(b)=b$ for $b\in\R$, which means that $f$ is the identity function. Hence $h_{p_1}=h_{p_2}$ for all $p_1, p_2\in (0,1)$. Let $h=h_{1/2}$. Hence, we have for all $p\in (0,1)$, $q_p^T=h\circ \widehat{\rho}_p$ with $h\in \mathcal F_U^\diamond$.

 Note that $q_p^T(\delta_x)=h\circ \widehat{\rho}_p(\delta_x)=h(x)$ for all $p\in (0,1)$. Therefore, $T(\delta_x)=\delta_{h(x)} $ for $x\in\R$. Using the fact that $T$ commutes with $\mathcal U$, we have $T^{\psi}\circ T(\delta_x)=T\circ T^u(\delta_x)$, implying $\psi\circ h=h\circ u$, and yielding $\psi=h\circ u \circ h^{-1}$. It follows from $T^{\psi}\circ T=T\circ T^u$ that
$T^{h}\circ T^u\circ T^{h^{-1}}\circ T=T\circ T^u$, implying
  $T^u\circ (T^{h^{-1}}\circ T)=(T^{h^{-1}}\circ T)\circ T^u$. Letting $\widehat{T}=T^{h^{-1}}\circ T$, we have $\widehat{T}$ commutes with $T^u$ for all $u\in\mathcal F_U$. In light of Theorem \ref{PTtheorem1}, we conclude that there exists $d\in\mathcal F_D$ such that $\widehat{T}=T_d$. Hence $T=T^{h}\circ T_d=T_d\circ T^{h}$.\end{proof}
  %\section{Proof of Theorem \ref{PTtheorem:4}}
  \section{Proof of Theorem \ref{RDUM2}}
  \label{app:D}
  \begin{proof}[Proof of Theorem \ref{RDUM2}]
We first show the ``if part". Suppose there exist $u\in\mathcal F_U^L$ and $d\in\mathcal F_D^R$ such that $T=T_d\circ T^u$, where $d$ is strictly increasing and continuous.
Note that for $d\in \mathcal F_D^R$, $T_d(F)(x)=d(F(x)),~x\in\mathbb{R}$. Hence, for $d\in \mathcal F_D^R$ and $u\in \mathcal F_U^L$
$$T_d\circ T^u(F)(x)=d(F(u_R^{-1}(x)))=T^u\circ T_d(F)(x),~x\in \mathbb{R},$$
indicating that $T^u$ and $T_d$ commute.
For any $T_{d_1}\in \mathcal D^R$, let $d_2=d\circ d_1\circ d^{-1}$, where $d^{-1}$ is the inverse function of $d$. Clearly, $d_2\in\mathcal F_D^R$.
Therefore, we have
$$T_{d_2}\circ T=T_{d_2}\circ T_d\circ T^u=T_{d_2\circ d}\circ T^u=T_{d\circ d_1}\circ T^u=T_{d}\circ T_{d_1}\circ T^u=T_{d}\circ T^u \circ T_{d_1}=T\circ T_{d_1}.$$
Moreover, by letting $d_2=d^{-1}\circ d_1\circ d\in\mathcal F_D^R$, we have
$$T\circ T_{d_2}=T_d\circ T^u\circ T_{d_2}=T_d\circ T_{d_2}\circ T^u=T_{d\circ d_2}\circ T^u=T_{d_1\circ d}\circ T^u=T_{d_1}\circ T_d\circ T^u=T_{d_1}\circ T.$$

We next show the ``only if" part.  Let $B_{y,z}^{\alpha}=\alpha\delta_y+(1-\alpha)\delta_z$ with $y<z$ and $\alpha\in [0,1]$. We   consider two different scenarios.

First, we suppose there exist $y_0<z_0$, $x_0\in\R$ and $\alpha_0\in (0,1)$ such that $T(B_{y_0,z_0}^{\alpha_0})(x_0)\in (0,1)$. Let $g(\alpha)=T(B_{y_0,z_0}^{\alpha})(x_0),~\alpha\in [0,1]$. Note that monotonicity of $T$ implies that $g$ is an increasing function. Moreover, commutation with $\mathcal D^R$ means that   for any $d_1\in\mathcal F_D^R$, there exist $d_2, d_3\in\mathcal F_D^R$ such that
$T_{d_2}\circ T=T\circ T_{d_1}$ and $T\circ T_{d_3}=T_{d_1}\circ T$.
By the second equality, we have $T\circ T_{d_3}(B_{y_0,z_0}^{\alpha})(x_0)=T_{d_1}\circ T(B_{y_0,z_0}^{\alpha})(x_0)$, which can be rewritten as $g\circ d_3=d_1\circ g$. Using the fact that $g(\alpha_0)\in (0,1)$, we have
$[0,1]=\{d_1(g(\alpha_0)): d_1\in\mathcal F_D^R\}\subseteq \{g(\alpha):\alpha\in [0,1]\}$. It follows from the fact  $g(\alpha)\in [0,1]$ for $\alpha\in [0,1]$ that $\{g(\alpha):\alpha\in [0,1]\}=[0,1]$. Hence $g\in\mathcal F_D^R$ is continuous.

By the first equality, we have $T_{d_2}\circ T(B_{y_0,z_0}^{\alpha})(x_0)=T\circ T_{d_1}(B_{y_0,z_0}^{\alpha})(x_0)$, which can be simplified as $d_2\circ g=g\circ d_1$. We next show that $g$ is strictly increasing. Suppose by contradiction that there exist $0<  \alpha_1<\alpha_2<\alpha_3<1$ such that $g(\alpha_1)=g(\alpha_2)<g(\alpha_3)$ or $g(\alpha_1)<g(\alpha_2)=g(\alpha_3)$. For the former case, let $d_1\in\mathcal F_D$ satisfy $d_1(\alpha_1)=\alpha_2$ and $d_1(\alpha_2)=\alpha_3$. This implies $g\circ d_1(\alpha_1)\neq g\circ d_1(\alpha_2)$. For any $d_2\in\mathcal F_D^R$, we have $d_2\circ g(\alpha_1)=d_2\circ g(\alpha_2)$, which is contradicted by the existence of $d_2\in\mathcal F_D^R$ such that $d_2\circ g=g\circ d_1$.  For the latter case, we can   show a similar contradiction. Hence $g$ is strictly increasing. Therefore, we conclude that $g\in\mathcal F_D^R$ is strictly increasing and continuous.

By $d_2\circ g=g\circ d_1$, we have $d_2=g\circ d_1\circ g^{-1}$, where $g^{-1}$ is the inverse function of $g$. It follows from $T_{d_2}\circ T=T\circ T_{d_1}$ that $T_{g\circ d_1\circ g^{-1}}\circ T=T_g\circ T_{d_1}\circ T_{g^{-1}}\circ T=T\circ T_{d_1}$. We denote $\widehat{T}=T_{g^{-1}}\circ T$.  It follows that $T_{d_1}\circ \widehat{T}=\widehat{T}\circ T_{d_1}$ for all $T_{d_1}\in \mathcal D^R$. This implies \eqref{Eq:utility} in the proof of Theorem \ref{PTtheorem2}. Hence, we have $\widehat{T}=T^u$ for some $u\in\mathcal F_U^L$, which further implies $T=T_g\circ T^u$ for some $u\in\mathcal F_U^L$ and some strictly increasing and continuous $g\in\mathcal F_D^R$.

Next, we consider the case that  $T(B_{y,z}^{\alpha})(x)\in\{0,1\}$ for all $y<z$, $x\in\R$ and $\alpha\in (0,1)$. Using
$T\circ T_{d_3}(\delta_y)=T_{d_1}\circ T(\delta_y)$, we have $T(\delta_y)=T_{d_1}\circ T(\delta_y)$ for any $d_1\in\mathcal F_D^R$. This implies $T(\delta_y)=\delta_{u(y)}$ for some function $u$. Hence for $\alpha=0, 1$, $T(B_{y,z}^{\alpha})(x)\in\{0,1\}$ also holds.
 We next show that for fixed $x,y,z$, $T(B_{y,z}^{\alpha})(x)=0$ for all $\alpha\in [0,1]$ or $T(B_{y,z}^{\alpha})(x)=1$ for all $\alpha\in [0,1]$. Suppose by contradiction that there exist $0<\alpha_1<\alpha_2\leq 1$ such that $T(B_{y,z}^{\alpha_1})(x)=0$ and $T(B_{y,z}^{\alpha_2})(x)=1$. Let $d_1\in\mathcal F_D^R$ such that $d_1(\alpha_1)=\alpha_2$. Then we have $T\circ T_{d_1}(B_{y,z}^{\alpha_1})(x)=T(B_{y,z}^{\alpha_2})(x)=1$.
However, for any $d_2\in\mathcal F_D^R$, $T_{d_2}\circ T(B_{y,z}^{\alpha_1})(x)=0$, which contradicts the existence of $d_2\in\mathcal F_D^R$ such that $T_{d_2}\circ T=T\circ T_{d_1}$. For $0\leq \alpha_1<\alpha_2<1$ such that $T(B_{y,z}^{\alpha_1})(x)=0$ and $T(B_{y,z}^{\alpha_2})(x)=1$, we can similarly show a contradiction. Hence, $T(B_{y,z}^{\alpha})(x)=0$ for all $\alpha\in [0,1]$ or $T(B_{y,z}^{\alpha})(x)=1$ for all $\alpha\in [0,1]$. This implies $T(B_{y,z}^{0})=T(B_{y,z}^{1})$ for all $y<z$, i.e., $T(\delta_y)=T(\delta_z)$ for all $y<z$.  Using the fact $T(\delta_y)=\delta_{u(y)}$, we have $T(\delta_y)=\delta_{c}$ for all $y\in\R$ and some $c\in\R$, which
together with   monotonicity of $T$ implies   $T=T^c$. Hence, we have $T=T_d\circ T^c$, where $d\in\mathcal F_D^R$ is strictly increasing and continuous. This completes the proof.
\end{proof}
\section{General sets of utility and distortion functions}
\label{app:general}
In this appendix, we study   probability distortions and utility transforms on different sets of UF and DF from the ones in Theorems \ref{PTtheorem1}, \ref{PTtheorem2} and \ref{RDUM}.
We use the notation introduced in Section \ref{Sec:RDU}.
The main intuition is that, for the commutation property, requiring more continuity in the distortion functions results in less continuity in the utility functions,
 and the same holds true if the positions of distortion functions and utility functions are switched. Therefore, the next result, which is a different version of Theorems \ref{PTtheorem1} and \ref{PTtheorem2}, further illustrates the symmetry between utility transforms and probability distortions.

%Let $\mathcal F_U^L$ be the set of all increasing and  left-continuous functions on $\R$ and let $\mathcal F_D^R$ be the set of all right-continuous functions  in $\mathcal F_D$. Note that  $\mathcal F_U\subseteq\mathcal F_U^L$ and $\mathcal F_D^R\subseteq\mathcal F_D$. The utility transform can be naturally extended to $\mathcal F_U^L$. We denote $\mathcal{U}^L=\{T^{u}:u \in \mathcal F_U^L\}$ and $\mathcal D^R=\{T_d: d\in \mathcal F_D^R\}$. Note that for $d\in \mathcal F_D^R$, $T_d(F)(x)=d\circ F(x),~x\in\R$.
 \begin{proposition}\label{PTtheorem4}
For a mapping $T: \M\to\M$,
\begin{enumerate}[(i)]
\item $T$ commutes
with each element of $\mathcal{U}^L$ if and only if $T\in \mathcal D^R$;
\item $T$ commutes with each element of  $\mathcal{D}^R$ if and only if $T\in\mathcal{U}^L$.
\end{enumerate}
  \end{proposition}
\begin{proof}
First note that the ``if" parts of (i)-(ii) are implied by the commutation of $T_d$ and $T^u$, which has been shown in the proof of Theorem \ref{RDUM2}.

 We next focus on the ``only if" part of (i). Note that the commutation of $T$ and each element in $\mathcal{U}^L$ implies the commutation of $T$ and each element in $\mathcal{U}$. By Theorem \ref{PTtheorem1}, we have $T=T_d$ for some $d\in\mathcal F_D$. Next we show  $d\in \mathcal F_D^R$ by contradiction. Suppose that there exists $x_0\in [0,1)$ such that $d(x_0)<d(x_0+)$. Moreover,  let $F_U$ be the standard uniform distribution on $[0,1]$, and
 \begin{equation*}\label{eq:phi}u(x)=\left\{\begin{array}{cc}
x,&x\leq x_0\\
x+1,& x>x_0\\
\end{array}
\right..
\end{equation*}
Then we have $u\in\mathcal F_U^L$. Direct computation gives
  $$T_d\circ T^u(F_U)(x_0)=[d\circ(F_U\circ u_R^{-1})](x_0+)=\lim_{y\downarrow x_0}d(F_U(u_R^{-1}(y)))=d(x_0),$$
 and $$T^u\circ T_d (F_U)(x_0)=w(u_R^{-1}(x_0))=w(x_0)=d(x_0+), $$
 where $w(x)=(d \circ F_U)(x+)=\lim_{y\downarrow x}d(F_U(y))$. This implies $T_d\circ T^u\neq T^u\circ T_d$, leading to a contradiction. Hence, $d\in\mathcal F_D^R$, and we establish the ``only if" part of (i).

%Hence we have
%$$d(1)=\lim_{x\to\infty}T(\delta_0)(x)=1,~\text{and}~d(0)=\lim_{x\to-\infty}T(\delta_0)(x)=0,$$
%implying $d\big{|}_{[0,1]}\in \mathcal F_D^R$.

The ``only if" part of (ii) is indicated by \eqref{Eq:utility} in the proof of Theorem \ref{PTtheorem2}.
\end{proof}
%  \begin{theorem}
%For a mapping $T: \M\to\M$,
%  $T$ commutes with each element in  $\mathcal{D}^R$ if and only if $T\in\mathcal{U}^L$.
%  \end{theorem}

\section{General spaces of distributions}\label{Sec:AD}

Let $\M_0$ be the set of all distributions on $\mathbb R$. In this appendix, we consider the distributional transforms $T:\M_0\to\M_0$. Both probability distortions and utility transforms are naturally extended from $\M$ to $\M_0$.
We obtain the results of Theorems \ref{PTtheorem1}-\ref{PTtheorem2} and Proposition \ref{PTtheorem4} for distributional transforms defined on $\M_0$. %In this section, we only consider the distributional transforms $T:\M_0\to\M_0$.
Let $\widehat{\mathcal F_D}$ be the set of all increasing functions $d:[0,1]\to [0,1]$ with $d(0)=d(0+)=0$ and $d(1)=d(1-)=1$. We denote $\widehat{\mathcal D}=\{T_d: d\in \widehat{\mathcal F_D}\}$.  In the next proposition, we derive some characterizations of probability distortions and utility transforms defined on $\M_0$.
\begin{proposition}\label{PTtheorem5}
For a mapping  $T: \M_0\to\M_0$,
 \begin{enumerate}
  \item[(i)] $T$ commutes with each element of $\mathcal U$ if and only if $T\in\widehat{\mathcal D}$;
   \item[(ii)] $T$ commutes with each element of $\widehat{\mathcal D}$ if and only if $T\in \mathcal U$.
  \end{enumerate}
     \end{proposition}
  \begin{proof} The ``if" parts of both (i) and (ii) follow from the same arguments as in the proof of Theorem \ref{PTtheorem1}. We next focus on the ``only if" parts.

   (i) First note that $T\circ T^u=T^u\circ T$ for all bounded $u\in\mathcal F_U$ implies that $T(F)\in\M$ for all $F\in\M$. Hence, in light of Theorem \ref{PTtheorem1}, there exists $d\in\mathcal F_D$ such that for all $F\in \M$, $$T(F)(x)=(d \circ F)(x+),~x\in \mathbb{R}.$$
   We next show that this equality holds for all $F\in\M_0$.
  Let $u(x)=\arctan(x),~x\in\mathbb{R}$. Then for any $F\in \M_0$, we have $T^u(F)\in\M$. Hence for $F\in\M_0$
  $$ T\circ T^u(F)(x)=[d \circ(F\circ u_R^{-1})](x+)=\lim_{y\downarrow x}d(F(u_R^{-1}(y)))=w(u_R^{-1}(x)),$$
  and
  $ T^u\circ T(F)(x)=(T(F))(u_R^{-1}(x)),$
  where $w(x)=\lim_{y\downarrow x}d(F(y))$. Using the fact that $T\circ T^u=T^u\circ T$ and $u_R^{-1}((-\pi/2, \pi/2))=\mathbb{R},$  we have
  $T(F)(x)=w(x)=T_d(F)(x),~x\in\mathbb{R}.$
  Hence $T(F)=T_d(F)$ for all $F\in\M_0$.  Taking $F(x)= e^x \wedge 1,~x\in\mathbb{R}$, it follows that
  $$d(0+)=\lim_{x\to -\infty}d(e^x\wedge 1)=\lim_{x\to -\infty}T(F)(x)=0.$$
  Moreover, letting $F(x)= (1-e^{-x})\id_{\{x\geq 0\}},~x\in\mathbb{R}$, we analogously obtain $d(1-)=1$.
  Consequently, $d\in\widehat{\mathcal F_D}$. We establish claim (i).

    (ii) We denote $F_0(x)=\frac{1}{2}+\frac{1}{\pi}\arctan(x),~x\in\mathbb{R}$ and note that $F_0\in\M_0$.
  For $F\in\M_0$, letting  $d(x)=F((F_0)_R^{-1}(x)),~x\in (0,1)$ with $d(0)=0$ and $d(1)=1$, it follows that $d$ is right-continuous, $d\in \widehat{\mathcal F_D}$ and $T_d(F_0)=F$. Hence,
  $$\M_0=\{T_d(F_0): d~\text{is right-continuous and}~d\in \widehat{\mathcal F_D}\}.$$
  For $F=T_d(F_0)$ with right-continuous $d\in \widehat{\mathcal F_D}$, using $T\circ T_d=T_d\circ T$, we have for $p\in (0,1)$
  $$q_p^T(F)=q_p^T(T_d(F_0))=q_p^{T_d}(T(F_0))=q_{d_L^{-1}(p)}^T(F_0).$$
  Note that for $p\in (0,1)$, $d_L^{-1}(p)=F_0(F_L^{-1}(p))$ and $d_L^{-1}(p)\in (0,1)$. Letting $u(x)=q_{F_0(x)}^T(F_0)$, we have $$q_p^T(F)=u(F_L^{-1}(p)),~p\in (0,1),~\text{and}~u\in \mathcal F_U^L,$$ implying $T=T^{u}$. We can show the continuity of $u$ using the same argument as in the proof of Theroem \ref{PTtheorem2}.    %This completes the proof.
  \end{proof}

We next obtain a version of Proposition \ref{PTtheorem4} for distributional transforms defined on $\M_0$.

  \begin{proposition}\label{PTtheorem6}
For a mapping $T: \M_0\to\M_0$,
\begin{enumerate}[(i)]
\item $T$ commutes
with each element of $\mathcal{U}^L$ if and only if $T\in \mathcal D^R\cap \widehat{\mathcal D}$;
\item $T$ commutes with each element of  $\mathcal{D}^R\cap \widehat{\mathcal D}$ if and only if $T\in\mathcal{U}^L$.
\end{enumerate}
  \end{proposition}
  \begin{proof} The ``if" parts of both (i) and (ii) follow from the same arguments as in the proof of Proposition \ref{PTtheorem4}. We next focus on the ``only if" parts.

  (i) By Proposition \ref{PTtheorem5}, there exists $d\in\widehat{\mathcal F_D}$ such that for all $F\in \M_0$, $$T(F)(x)=(d \circ F)(x+),~x\in \mathbb{R}.$$
Using the same argument as in the proof of (i) of Proposition \ref{PTtheorem4}, we can show that $d$ is right-continuous. Hence $T=T_d\in \mathcal D^R\cap \widehat{\mathcal D}$.
 The ``only if" part of (ii) follows from the same argument as in the proof of (ii) of Proposition \ref{PTtheorem5}. Hence the detail is omitted.
  \end{proof}

 % \begin{remark}
 % Note that we only use one distribution $F_0$ together with $\mathcal F_D^R\cap \widehat{\mathcal F_D}$ in the proof of (i) of Theorem \ref{PTtheorem4} to represent $\M_0$. That is $\M_0=\{T_d(F_0): d\in\mathcal F_D^R \cap \widehat{\mathcal F_D}\}$. However, we could not find such a distribution in $\M$ to generate $\M$. Hence we use a sequence of distributions $F_{U_n}\in\M,~n\geq 1$, to generate $\M$ in the proof of (i) of Theorem \ref{PTtheorem2}.
 % \end{remark}

%With monotonicity, the set $\mathcal F_U$ in (i) of Theorem \ref{PTtheorem3} can be shrunk to $\mathcal F_U$.
%\begin{proposition}\label{Moncom} For  a distributional transform $T:\M_0\to\M_0$,
%$T$ is monotone and commutes with $T^u$ for all $u\in \mathcal F_U$ if and only if $T=T_d$ for some $d\in \mathcal F_D$.
%\end{proposition}
%\begin{proof}
%We first focus on the ``if part". Monotonicity follows from the check of definition and commutation with utility transforms follows from Theorem \ref{PTtheorem3} (i).   The ``only if" part is similar to the proof of Theorem \ref{PTtheorem1}. The only difference is that  we have to apply Theorem 1 of \cite{C09} to prove Equation \eqref{eq:quantile}.
%\end{proof}
We define the RDU transform on $\M_0$ as follows: a distributional transform $T:\M_0\to\M_0$ is an RDU transform if there exist $d\in\widehat{\mathcal F_D}$ and $u\in\mathcal F_U^L$ such that $T=T_d\circ T^u$.

With this definition of RDU,  the results in Theorems \ref{RDUM}-\ref{RDUM2} remain true for the mappings $T:\M_0\to\M_0$.
\begin{proposition} Let $T:\M_0\to\M_0$ be a monotone distributional transform. The following hold:
\begin{enumerate}[(i)]
\item $T$ commutes with $\mathcal U$ if and only if $T$ is an RDU transform with strictly increasing and surjective UF;
\item $T $ commutes with $\mathcal{D}^R\cap \widehat{\mathcal D}$ if and only if
    $T$ is an RDU transform with left-continuous UF and strictly increasing and continuous DF.
\end{enumerate}
 \end{proposition}
 \begin{proof} (i) First note that Proposition \ref{mappings} also holds for $\rho:\M_0\to\R$. All the rest of the proof follows the same argument as in the proof of Theorem \ref{RDUM}. We omit the details.

 (ii) The ``if" part follows from  checking the definition, which is exactly the same as the proof of Theorem \ref{RDUM2}. We next show the ``only if" part, which only requires some refinement of the proof of Theorem \ref{RDUM2}. Along with the proof of Theorem \ref{RDUM2}, we consider two different scenarios. Recall that $B_{y,z}^{\alpha}=\alpha\delta_y+(1-\alpha)\delta_z$ with $y<z$ and $\alpha\in [0,1]$. The proof of the case that $T(B_{y_0,z_0}^{\alpha_0})(x_0)\in (0,1)$ for some $y_0<z_0$, $x_0\in\R$ and $\alpha_0\in (0,1)$ follows the similar argument as that of Theorem \ref{RDUM2}. We next consider the case that $T(B_{y,z}^{\alpha})(x)\in\{0,1\}$ for all $y<z$, $x\in\R$ and $\alpha\in (0,1)$. Following the same argument as in the proof of Theorem \ref{RDUM2}, we have that for some $c\in\R$, $T(F)=\delta_c$ for all $F\in\mathcal M$. We aim to extend this conclusion to $\M_0$. Note that commutation with $\mathcal{D}^R\cap \widehat{\mathcal D}$ implies that for any $d_1\in\mathcal F_D^{R}\cap \widehat{\mathcal{F}_D}$, there exist $d_2, d_3\in\mathcal F_D^R\cap \widehat{\mathcal{F}_D}$ such that
$T_{d_2}\circ T=T\circ T_{d_1}$ and $T\circ T_{d_3}=T_{d_1}\circ T$. Let $d_1$ be a function in $\mathcal F_D^{R}\cap \widehat{\mathcal{F}_D}$ additionally satisfying $d_1(\epsilon)=0$ and $d_1(1-\epsilon)=1$ for some $\epsilon\in (0,1/2)$. This implies $T_{d_1}(F)\in\M$ for all $F\in\M_0$. It follows from the conclusion on $\M$ that
$T_{d_2}\circ T(F)=T\circ T_{d_1}(F)=\delta_c.$ Hence, we have  for all $F\in\M_0$, $d_2(T(F)(c))=1$ and $d_2(T(F)(x))=0$ for all $x<c$. Moreover, by $T\circ T_{d_3}=T_{d_1}\circ T$, we have  $T(d_3\circ F)(c)=d_1(T(F)(c))$. By freely choosing $d_1$, we have  that if $T(F)(c)<1$, then there exists $d_3$ such that $T(d_3\circ F)(c)=d_1(T(F)(c))=0$. This implies $d_2(T(d_3\circ F)(c))=0$, which contradicts $d_2(T(F)(c))=1$ for all $F\in\M_0$. Hence $T(F)(c)=1$ for all $F\in\M_0$. If $T(F)(x)>0$ for some $x<c$, then  by freely choosing $d_1$ we have that there exists some $d_3$ such that $T(d_3\circ F)(x)=d_1(T(F)(x))=1$. Then we have $d_2(T(d_3\circ F)(x))=1$, leading to a contradiction. Hence, $T(F)(x)=0$ for all $x<c$ and $F\in\M_0$. Combing the above conclusions, we obtain $T(F)=\delta_c$ for all $F\in\M_0$, which implies $T=T_d\circ T^c$, where $d\in \mathcal F_D^{R}\cap \widehat{\mathcal{F}_D}$ is strictly increasing and continuous. 
 \end{proof}
 %(ii) The ``if" part follows from the proof of Theorem \ref{RDUM2}. For the ``only if" part, the case of $T(B_{y_0,z_0}^{\alpha_0})(x_0)\in (0,1)$ for some $y_0<z_0$, $x_0\in\R$ and $\alpha_0\in (0,1)$, follows exactly the same argument as in the proof of Theorem \ref{RDUM2}. For the scenario of $T(B_{y,z}^{\alpha})(x)\in \{0,1\}$ for all $y<z$, $x\in\R$ and $\alpha\in (0,1)$, we could similarly show that for all $x\in\R$, $T(\delta_x)=\delta_c$ for some $c\in\R$. Using the monotonicity of $T$, we have $T(F)=\delta_c$ for all $F\in\M$. Next, we extend this conclusion to $\M_0$.  We denote $F_0(x)=\frac{1}{2}+\frac{1}{\pi}\arctan(x),~x\in\mathbb{R}$ and note that $F_0\in\M_0$.
 % For $F\in\M_0$, letting  $d(x)=F((F_0)_R^{-1}(x)),~x\in (0,1)$ with $d(0)=0$ and $d(1)=1$, it follows that $d$ is right-continuous, $d\in \widehat{\mathcal F_D}$ and $T_d(F_0)=F$.

 %Note that in Proposition \ref{Moncom} if $\M=\M_0$, $T:\M_0\to\M_0$ forces $d\in \mathcal F_D^{R}$ as $d\in\mathcal F_D\setminus \mathcal F_D^{**}$ may lead to $T(F)\notin\M_0$ for some $F\in\M_0$.

\end{document}